\documentclass[12pt]{article}
\usepackage{amsmath,amssymb,amsthm}
\usepackage{graphicx}
\usepackage{natbib}
\usepackage{bm}
\usepackage{booktabs}
\usepackage{multirow}
\usepackage{subcaption}
\usepackage{dsfont}
\usepackage{xcolor}

\usepackage{pgfplots}
\pgfplotsset{compat=1.16}
\usepackage{pgfplotstable}

\usepackage{algorithm,caption}
\usepackage[noend]{algpseudocode}

\newtheorem{theorem}{Theorem}
\newtheorem{lemma}{Lemma}
\newtheorem{prop}{Proposition}
\newenvironment{definition}[1][Definition]{\begin{trivlist}
\item[\hskip \labelsep {\bfseries #1}]}{\end{trivlist}}
\newenvironment{assumption}[1][Assumption 1]{\begin{trivlist}
\item[\hskip \labelsep {\bfseries #1}]}{\end{trivlist}}

\newcommand{\blind}{1}

\begin{document}

\def\spacingset#1{\renewcommand{\baselinestretch}%
{#1}\small\normalsize} \spacingset{1}

%%%%%%%%%%%%%%%%%%%%%%%%%%%%%%%%%%%%%%%%%%%%%%%%%%%%%%%%%%%%%%%%%%%%%%%%%%%%%%

\if1\blind
{
  \title{\bf Pathway Testing in Metabolomics with Globaltest, Allowing Post Hoc Choice of Pathways}
  \author{Ningning Xu\thanks{
    The authors acknowledge NWO VIDI grant number 639.072.412.}\hspace{.2cm}\\
    Department of Biomedical Data Sciences, Leiden University Medical Center\\
    and \\
    Aldo Solari \\
    Department of Economics, Management and Statistics, University of Milano-Bicocca \\
    and \\
    Jelle Goeman \\
    Department of Biomedical Data Sciences, Leiden University Medical Center}
  \maketitle
} \fi

\if0\blind
{
  \bigskip
  \bigskip
  \bigskip
  \begin{center}
    {\LARGE\bf Pathway Testing in Metabolomics with Globaltest, Allowing Post Hoc Choice of Pathways}
\end{center}
  \medskip
} \fi

\bigskip
\begin{abstract}
The Globaltest is a powerful test for the global null hypothesis that there is no association between a group of features and a response of interest, which is popular in pathway testing in metabolomics. Evaluating multiple pathways, however, requires multiple testing correction. In this paper, we propose a multiple testing method, based on closed testing, specifically designed for the Globaltest. The proposed method controls the family-wise error rate simultaneously over all possible feature sets, and therefore allows post hoc inference, i.e.\ the researcher may choose the pathway database after seeing the data without jeopardizing error control. To circumvent the exponential computation time of closed testing, we derive a novel shortcut that allows exact closed testing to be performed on the scale of metabolomics data. An \texttt{R} package \texttt{ctgt} is available on CRAN. We illustrate the shortcut on several metabolomics data examples.
\end{abstract}

\noindent%
\vfill

\newpage
\spacingset{1.5} % DON'T change the spacing!
\section{Introduction}
\label{sec:intro}

In high-dimensional data, features may often be meaningfully taken together in sets or groups. It has been suggested in genomics that pathways, sets of biologically related genes, can better reflect the underlying biological events than single-SNP analysis \citep{wang2010analysing}. This is especially true in metabolomics, where each metabolic pathway is a group of functionally-associated metabolites. Analyzing metabolomics data in terms of such pathways provides mechanistic insights into the underlying biology \citep{xia2015metaboanalyst}. 

Many methods can be used for pathway testing in metabolomics, for instance ``GSEA'' \citep{subramanian2005gene}, ``PAGE'' \citep{kim2005page} and ``SAM-GS'' \citep{dinu2007improving}. One very popular method is the Globaltest \citep{goeman2004global}, originally proposed in the context of gene set testing in transcriptomics. \cite{goeman2006testing} proved that Globaltest is locally most powerful on average in a neighborhood of the null hypothesis and remains valid and powerful in high-dimensional data with more features than observations. Furthermore, the method adapts to the correlation structure in the data. In MetaboAnalyst \citep{xia2015metaboanalyst}, a web-based analytical pipeline for metabolomics data, Globaltest is the default testing approach for pathway analysis.

When many pathways are tested, multiple testing correction methods are necessary. Researchers may either control the family-wise error rate (FWER), the probability of at least one false rejection, or the false discovery rate (FDR), the expected proportion of false rejections among all rejections. We follow \cite{meijer2016multiple}, who argued that FWER is more appropriate for pathway testing than FDR, which can be difficult to interpret when pathways are nested. To control FWER, several multiple testing methods have been proposed combining Globaltest with Bonferroni, such as the Focus Level procedure \citep{Goeman2008multiple}, DAG \citep{Meijer2015dag}, and Structured Holm \citep{meijer2016multiple}. All these methods control the error rate over a collection of pathways that must be specified before the data are seen. When there are numerous pathways with large overlaps, Bonferroni-based methods may be overly conservative.

There have been a surge of specialized databases for metabolic pathways in the past decade \citep{likic2006databases}, such as KEGG (Kyoto Encyclopedia of Genes and Genomes) \citep{kanehisa2002kegg}, WikiPathways \citep{slenter2017wikipathways}, BioCyc \citep{karp2019biocyc} and SMPDB (The Small Molecule Pathway Database) \citep{frolkis2010smpdb}. 
While the information in these pathway databases overlaps, the databases differ markedly from each other. Each contains many unique pathways not present in other databases. Even when the same pathway is present in multiple databases, the metabolites associated with the pathway in each database can be quite different \citep{mubeen2019impact}. The abundance of pathway information in all these databases is in principle a good thing, providing a wealth of information for the researcher to use. However it creates a conundrum when correcting for multiple testing, since current methods require researchers to declare, before seeing the data, which pathway database(s) will be used. Researchers either have to restrict to a single database, limiting the information available to them, or declare to use all databases and incur an unnecessarily large multiple testing penalty due to the numerous overlapping pathways. Too often, current practice, unfortunately, is to correct for multiple testing only within each database even when multiple databases have been used \citep{lopez2016mbrole}.

To illustrate the loss of FWER control when correcting for multiple testing only within database, we simulated null data by permuting the 0/1 response 2000 times based on a real data set with 92 observations and 47 metabolites \citep{taware2018volatilomic}. 
We consider a total of 6 annotation databases: ``KEGG", ``Biocyc", ``SMPDB", ``Biofunction'', ``Protein'' from Metabolites Biological Role (MBROLE) \citep{lopez2016mbrole} and ``Wiki'' from ``rWikiPathways'' \citep{slenter2017wikipathways}. 
Table \ref{tab:FWER} shows that DAG, Focus Level and Structured Holm control the FWER at the pre-specified level of $5\%$ within each database, as expected. However, when considering all databases simultaneously, an error corresponds to rejecting at least one true null hypothesis in any database. In this case, DAG, Focus Level and Structured Holm have a FWER of  11.9\%, 12.3\% and 8.3\%, respectively, largely exceeding the $5\%$ level.

\begin{table}[!htbp]
\caption{FWER for different methods (DAG, Focus Level and Structured Holm at level $5\%$) per database and simultaneously over all databases.} \label{tab:FWER}
\centering
    \begin{tabular}{r|cccccc|r}
     \toprule
Database             & KEGG    & Biocyc   &   SMPDB   &   Biofunction & Protein &   Wiki   &  Overall\\
 No. of pathways     & 16    & 18       &  4        &   8          & 12      &    7      &  65\\
        \toprule
DAG          & 3.9\%      & 3.5\%    & 4.5\%     & 3.4\%              & 4.4\%             &   3.0\%       & 11.9\% \\ 
Focus Level  & 3.8\%      & 3.6\%      & 4.5\%     & 3.7\%              & 3.8\%             &  3.1\%      & 12.3\% \\ 
Structured Holm & 2.3\% & 2.9\%    & 3.2\%       & 3.6\%                & 2.7\%               &  2.9\%        & 8.3\% \\ 
   \bottomrule
  \end{tabular}
\end{table}

Ideally, researchers should be able to choose the pathway databases after seeing the data, while still maintaining strict FWER control. Such post hoc inference is possible using closed testing \citep{marcus1976closed}. Closed testing controls FWER for all possible feature sets, and therefore allows researchers to postpone the selection of pathways of interest after seeing the data. \cite{goeman2019only} proved that only closed testing procedures are admissible for FWER control of pathways, i.e.\ all other procedures either are equivalent to closed testing or can be improved using closed testing. Closed testing has been explored for pathway analysis in genomics by \cite[``SEA'',][]{mitra2019}, building upon applications in neuroimaging \citep{rosenblatt2018all}. However, ``SEA" uses the Simes' test \citep{simes1986improved} as the pathway test, which requires assumption on positive dependence of p-values \citep{sarkar1997simes} and is conservative when p-values are strongly dependent. Simes is an unusual pathway test to use in metabolomics, and it would be desirable to use Globaltest instead.

In this work, we develop closed testing for FWER control with Globaltest pathway analysis for application to metabolomics. The major challenge to perform closed testing is computational: it requires exponentially many tests. We develop novel shortcuts to overcome this limitation by reducing the exponential number of Globaltests to linear. We first propose a ``single-step'' shortcut that is fast but approximate to the full closed testing procedure. It guarantees strong FWER control but may be conservative. To gain power, we then embed the single-step shortcut within a branch and bound algorithm, leading to an ``iterative'' shortcut. The iterative shortcut will approximate the full closed testing procedure closer and closer as we iterate longer, trading computation time for power and converging eventually to the exact closed testing result. On the scale of typical metabolomics data ($\approx$ 300 features), the exact closed testing result for a pathway can be obtained in seconds on a regular PC.

Although Globaltest is derived in the context of all generalized linear models we focus in this paper on logistic regression only, which is the most popular generalized linear model used with Globaltest. We first revisit Globaltest and its properties in Section \ref{sec:gt} and the closed testing procedure in Section \ref{sec:ct}. We introduce the main idea of the single-step shortcut in Section \ref{sec:sss} and the iterative extension in Section \ref{sec:it}. In the remaining sections, we explore applications of our method on real metabolomics datasets.

\section{The Globaltest}\label{sec:gt}

The data generated in metabolomics usually consist of measurements performed on subjects under different conditions. Suppose we have $n$ independent subjects on which we have measured hundreds of metabolite variables and a $0/1$ response representing two different conditions. We gather the $n$ observations into a response vector $\mathbf{y}$ and a design matrix partitioned into an $n\times m$ matrix $\mathbf{X}$ of metabolites and an $n\times z$ matrix $\mathbf{Z}$ including the intercept term and potential confounders we would like to adjust for, e.g. age and gender. As we are interested in large-scale metabolomics, we allow the number of metabolites to be larger than the number of subjects, i.e. $m > n$, although we assume that $z < n$.

To denote a pathway we will use the index set $R \subset \{1,\ldots,m\}$ of the metabolites it includes, and we will write $r=|R|$ for its cardinality and $\mathbf{X}_R$ for the submatrix of $\mathbf{X}$ formed from columns indexed by $R$. Globaltest \citep{goeman2004global, goeman2006testing, goeman2011testing} is a powerful testing method to test pathways for association with the response, especially in the case that many metabolites in the pathway are associated with the response in a small way. The Globaltest assumes the logistic model
\begin{equation}\label{eq:log}
\mathbb{E}(\mathbf{y} \mid \mathbf{Z}, \mathbf{X}_R) =  h( \mathbf{Z} \bm{\gamma} + \mathbf{X}_R \bm{\beta}), 
\end{equation}
to test the null hypothesis
$$\bm{\beta} = \mathbf{0}$$
for a $r$-dimensional vector $\bm{\beta}$, where $h(t) = \exp(t)/(1+\exp(t))$ is the standard logistic function and $\bm{\gamma}$ is a $z$-dimensional vector of nuisance parameters. 

The Globaltest statistic for testing $\bm{\beta} = \mathbf{0}$ in the model (\ref{eq:log}) is given by
\begin{equation}\label{eq:gt}
g_R = \mathbf{y}^\intercal(\mathbf{I}-\mathbf{H}) \mathbf{X}_R \mathbf{X}_R^\intercal (\mathbf{I}-\mathbf{H}) \mathbf{y}
\end{equation}
where $\mathbf{I}$ is the identity matrix of size $n$ and $\mathbf{H} = \mathbf{Z}(\mathbf{Z}^\intercal \mathbf{Z})^{-1}\mathbf{Z}^\intercal$. 
It can be seen from \eqref{eq:gt} that the Globaltest statistic is the sum of the test statistics calculated for the single metabolites in $R$, that is, $g_R = \sum_{i\in R} g_i$, where $g_i = \mathbf{y}^\intercal(\mathbf{I}-\mathbf{H}) \mathbf{X}_i \mathbf{X}_i^\intercal (\mathbf{I}-\mathbf{H}) \mathbf{y}$.

Theorem 1 in \cite{goeman2011testing} shows that the null distribution of $g_R$ is asymptotically equivalent to a weighted sum of independent $\chi^2_1$ variables, i.e.\
\begin{equation}\label{eq:asygt}
\sum_{i=1}^n \lambda_i^R \chi^2_1
\end{equation} 
where the weights $\lambda_1^R\geq  \cdots\geq  \lambda_n^R$ are the eigenvalues of the positive semidefinite matrix $\bm{\Sigma}^{1/2}(\mathbf{I}-\mathbf{H}) \mathbf{X}_{R} \mathbf{X}^\intercal_{R} (\mathbf{I}-\mathbf{H}) \bm{\Sigma}^{1/2}$ in descending order. Here $\bm{\Sigma}$ is the diagonal covarance matrix of $\mathbf{y}$ under the null hypothesis, with entries $\mathbb{E}(\mathbf{y} \mid \mathbf{Z}) (1-\mathbb{E}(\mathbf{y} \mid \mathbf{Z})) $.

For a pre-specified significance level $\alpha$, we can approximate the Globaltest critical value 
\begin{equation*}\label{eq:cv}
 c_R = c(\bm{\lambda}^R)
\end{equation*}
by the $1-\alpha$ quantile of the asymptotic null distribution, where we have made the dependence of $c_R$ on the eigenvalues $\bm{\lambda}^R = (\lambda_1^R, \cdots, \lambda_n^R)$ explicit. To compute the critical value $c_R$, \cite{goeman2011testing} suggested to use the algorithms proposed by \cite{Imhof1961} and \cite{Robbins1949}. Robbins' algorithm, though slightly slower on average, is numerically more stable and less vulnerable to the problem of premature convergence.

The following Proposition \ref{prop:asygt} shows that the Globaltest $\phi_R = \mathds{1}\{g_R \geq c_R\}$ is an asymptotically valid $\alpha$-level test. The proof is in the supplementary material.
\begin{prop}\label{prop:asygt}
Assume that the logistic model in (\ref{eq:log}) holds with $\bm{\beta}=0$, then
$$\lim_{n\rightarrow \infty} \mathbb{E}(\phi_R) \leq \alpha, $$
that is, the Globaltest has asymptotic type I error control. 
\end{prop}

\section{Closed Testing}\label{sec:ct}

When testing pathways, we are interested in finding pathways in which there is evidence of some association between the signal of the metabolites in the pathway and the response. We suppose some metabolites are associated with the response and some metabolites are not associated with the response, i.e.\ \emph{null metabolites}. As usual with Globaltest, we adopt the self-contained paradigm for pathway testing \citep{goeman2007analyzing}, in which a null pathway is defined as a pathway containing only null metabolites. Let $F = \{1,\cdots,m\}$ be the set of all metabolites, and $N \subseteq F$ the set of all null metabolites. For any pathway $R \subseteq F$, the self-contained null hypothesis of pathway $R$ is
\begin{equation}\label{eq:hr}
H_R: R\subseteq N.
\end{equation}

To allow post hoc inference we will control FWER for the family $\mathcal{F} = 2^F$ of all $2^m$ possible pathways. The collection of null pathways is $\mathcal{N} = \{I \in \mathcal{F}\colon I \subseteq N\}$. Our goal is to design a test procedure that rejects the collection of pathways $\mathcal{X} \subseteq \mathcal{F}$ in such a way that FWER is controlled, i.e.\
\begin{equation} \label{fwer}
\mathrm{Pr} (\mathcal{X} \cap \mathcal{N} \neq \emptyset) \leq \alpha.
\end{equation}

To obtain such FWER control we will use the closed testing procedure \citep{marcus1976closed}. Closed testing requires that the family of hypotheses is closed under intersection: for all $H_A$, $H_B$ in the family we should have $H_A \cap H_B$ in the family. This is easy to check for the hypothesis family $\{H_I\colon I \in \mathcal{F}\}$, since $H_A \cap H_B = H_{A\cup B}$. A null hypothesis $H_I$ is rejected by closed testing if and only if all the hypotheses that are implied by it have been rejected by a level $\alpha$ test. Formally, suppose for every $I \in F$, $\psi_I$ is a test of $H_I$ with 1 indicating rejection and 0 non-rejection. Closed testing rejects $\mathcal{X} = \{I \in \mathcal{F}\colon \psi^F_I=1\}$, where
\[
\psi^F_I = \min\{\psi_S\colon I \subseteq S \subseteq F\}. 
\]
The closed testing has FWER control (\ref{fwer}) if $\psi_N$ is a valid $\alpha$-level test of $H_N$, i.e. when $\mathbb{E}(\psi_N)\leq \alpha$ \citep{marcus1976closed}. This generalizes to asymptotic FWER control if the test for $H_N$ is asymptotically valid, as we summarize in the following Proposition \ref{prop:fwer}; see supplementary material for the proof.
\begin{prop}\label{prop:fwer}
If $\displaystyle \lim_{n\to\infty} \mathbb{E}(\psi_N) \leq \alpha$, then $\displaystyle \lim_{n\to\infty} \mathrm{Pr} (\mathcal{X} \cap \mathcal{N} \neq \emptyset) \leq  \alpha$.
\end{prop}

Based on the discussion above, to be able to use closed testing with Globaltest, we need to assume that the Globaltest $\phi_N$ is an asymptotically valid $\alpha$-level test of $H_N$. We thus assume that the logistic model holds for model $N$, i.e.\ 
\begin{assumption}\label{asp}
$\mathbb{E}(\mathbf{y} \mid \mathbf{Z}, \mathbf{X}_N)  =  h( \mathbf{Z} \bm{\gamma}) $.  
\end{assumption}

Under this assumption, Globaltest for $H_N$ is an asymptotically valid $\alpha$ level test, based on Proposition \ref{prop:asygt}, and consequently FWER control (\ref{fwer}) applies by Proposition \ref{prop:fwer}. We note that we only need to assume the correct model specification for one single logistic regression model, i.e. model $N$. This is important, since it is not generally possible for several nested logistic models to be simultaneously valid \citep{gail1984biased}. This robustness to model misspecification is a useful and often overlooked property of closed testing.

\section{Single-step Shortcut}\label{sec:sss}

A hypothesis $H_R$ can be rejected by closed testing if all hypotheses $H_S$ with $ R \subseteq S \subseteq F$ are rejected by the Globaltest at level $\alpha$. However, this results in exponential computational complexity of closed testing, which can be problematic for large-scale metabolomics data sets. Shortcuts, efficient algorithms, are thus necessary to reduce computation time \citep{brannath2010shortcuts,gou2014class,dobriban2018flexible}. Shortcuts can be exact or approximate. Approximate shortcuts control FWER, but sacrifice power relative to the full closed testing procedure. In this paper, we first derive an approximate single-step shortcut and subsequently an exact iterative shortcut for closed testing with Globaltest. We start with the single-step shortcut.

\subsection{Main Idea}
For any set $R$ of interest, closed testing rejects $H_R$ with asymptotic FWER control at level $\alpha$ if and only if \[g_S \geq c_S, \text{ for all } R\subseteq S \subseteq F. \] This means that naively we have to calculate test statistics and critical values of Globaltest for a total of $2^{\vert F\vert - \vert R\vert}$ hypotheses. 

For illustration, we use a recurring toy example with $n=100$ observations, $m=5$ features and a binary response. Let $F = \{1,2,3,4,5\}$ be the index set of all features. Suppose that we want to test whether $R = \{3\}$ is associated with the response or not. By closed testing, we have to calculate all test statistics $g_S$ and critical values $c_S$ for testing all $H_S$ with $R \subseteq S \subseteq F$. All $g_S$ and $c_S$ are presented in Figure \ref{fig:tc} by circles and triangles respectively. For each pair of $(g_S, c_S)$, if circles are above triangles, closed testing then rejects $H_R$.

\begin{figure}[!ht]
\centering
\includegraphics[scale=0.68]{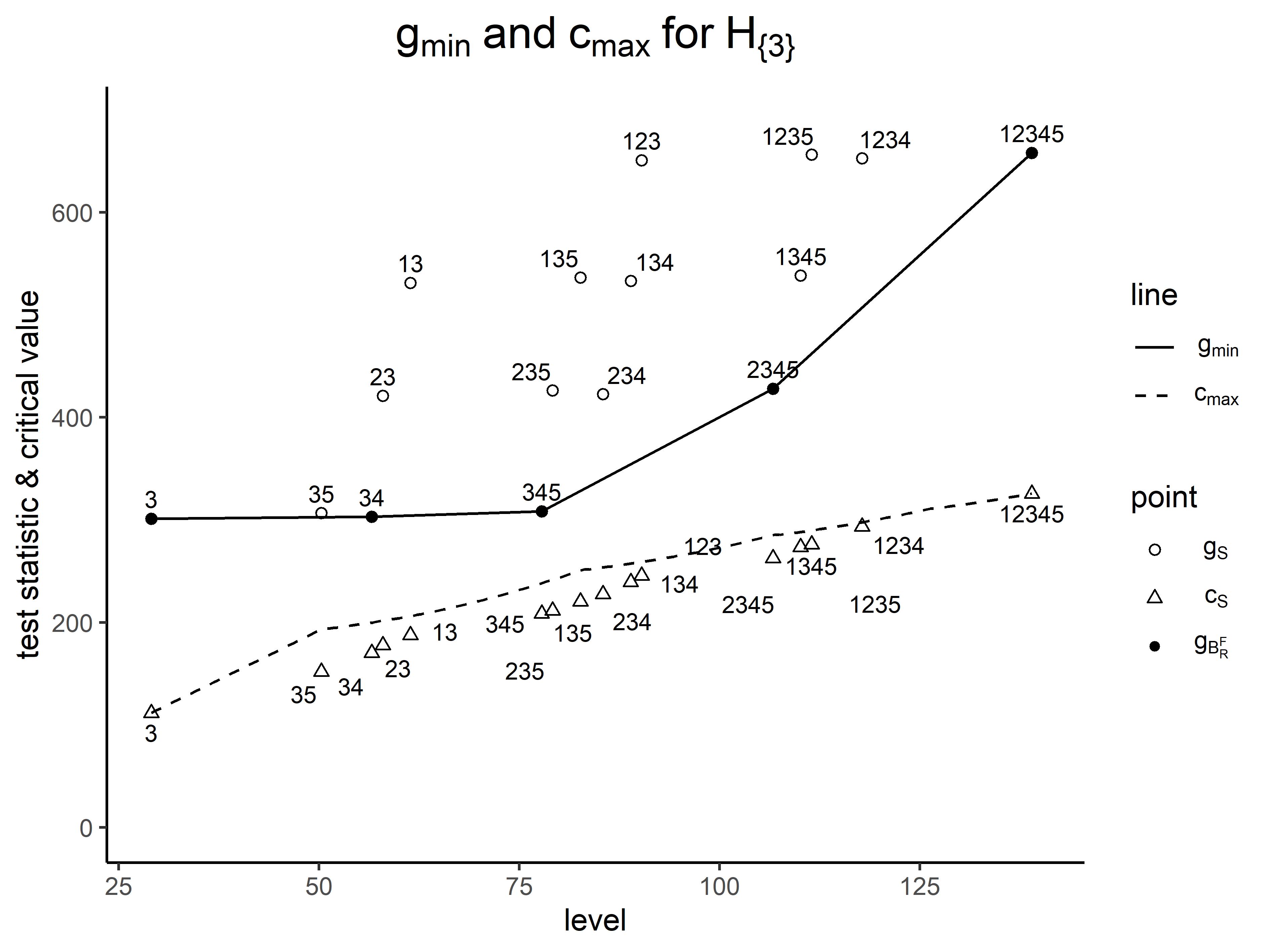}
\caption{Single-step shortcut for testing $H_{\{3\}}$. Circles and triangles denote test statistics and critical values, respectively, for all $H_S$ with $R \subseteq S \subseteq F$. The solid line represents the minimum test statistic $g_{min}(\ell)$ and the dashed line represents the maximal critical value $c_{max}(\ell)$.  }
\label{fig:tc}
\end{figure}

Define $\ell_S = \sum_{i=1}^{n}\lambda_i^S $ as the ``level'' of $H_S$. 
This definition facilitates the construction of the ``single-step'' shortcut since 
the comparison between $g_S$ and $c_S$ is restricted to the range $[\ell_R, \ell_F]$; see the x-axis in Figure \ref{fig:tc}.
The main idea of the single-step shortcut is as follows. We propose to construct a minimum test statistic line $g_{min}(\ell)$ and a maximal critical value line $c_{max}(\ell)$, both are monotonically increasing on $\ell\in [\ell_R, \ell_F]$ such that 
\begin{equation}\label{eq:gmin_i}
g_S \geq g_{min}(\ell_S) \tag{i}
\end{equation}
and
\begin{equation}\label{eq:cmax_ii}
c_S \leq c_{max}(\ell_S). \tag{ii}
\end{equation}
If such $g_{min}(\ell)$ and $c_{max}(\ell)$ can be established, we then simply compare the two lines instead of the exponentially many test statistics and critical values. When the $g_{min}$ line is above the $c_{max}$ line, $H_R$ is surely rejected by closed testing, as the following lemma says. 
\begin{lemma}\label{lm:ss}
If $g_{min}$ and $c_{max}$ satisfy \eqref{eq:gmin_i} and \eqref{eq:cmax_ii}, respectively, then closed testing rejects $H_R$ at level $\alpha$ when $g_{min}(\ell) \geq c_{max}(\ell)\,\, \forall \, \ell \in [\ell_R ,\ell_F].$
\end{lemma}
Proof of Lemma \ref{lm:ss} is in the supplemental file. In the following we will show how to construct $g_{min}(\ell)$ and $c_{max}(\ell)$.

\subsection{The Minimum Test Statistic}\label{sec:gmin}

We define the minimum test statistic $g_{min}(\ell), \ell \in [\ell_R, \ell_F]$ by using  the lower convex hull algorithm \citep{bronnimann2004space}. Given a set of points in the plane, the lower convex hull algorithm computes a line through the bottommost-leftmost point and the topmost-rightmost point (on the $x$-coordinate) via concatenating a certain in-between bottommost points such that all points are above or on the line. See for example the solid line in Figure \ref{fig:tc}, where the bottommost-leftmost point is $R= \{3\}$, the topmost-rightmost point is $F = \{1,2,3,4,5\}$ and the bottommost points are $\mathcal{B}^F_R = \{\{3\}, \{34\}, \{3,4,5\}, \{2,3,4,5\}, \{1,2,3,4,5\} \}$. 

By definition of lower convex hull, we have that
\begin{lemma}\label{lm:gmin}
$g_{min}(\ell)$ satisfies \eqref{eq:gmin_i} for all $S$ with $R\subseteq S \subseteq F$.
\end{lemma}
More details about the technical constructions of $g_{min}(\ell)$ are available in the supplementary file along with the proof of Lemma \ref{lm:gmin}. It is worth noting that constructing $g_{min}(\ell)$ in this way dramatically reduces the computational burden, since exponentially many calculations of the full closed testing is prevented. In the following, we will introduce the second half of the shortcut: how to compute the maximal critical values. 

\subsection{The Maximal Critical Value}\label{sec:cmax}

As discussed in Section \ref{sec:gt}, the critical value is a function of the eigenvalue vector, and so is the maximal critical value. We thus need to find out a numeric vector for which the corresponding critical value is maximal among all the critical values at the same level. To establish such a vector, we first introduce the definition of \textit{majorization} \citep{horn2012matrix}:
\begin{definition}\label{def:mj}
Let vectors $\bm{\lambda}=(\lambda_1,\cdots, \lambda_n)$ with $\lambda_1 \geq \cdots \geq \lambda_n$ and $\bm{\delta} = (\delta_1, \cdots,\delta_n)$ with $\delta_1 \geq \cdots \geq \delta_n$ be given. Then $\bm{\lambda}$ is said to majorize $\bm{\delta}$, i.e.\ $\bm{\lambda} \succ \bm{\delta}$ if $\sum_{i=1}^s \lambda_i \geq \sum_{i=1}^s \delta_i$ for all $s=1,\cdots,n$ with equality for $s=n$.
\end{definition}

By \textit{inclusion principle} for hermitian and positive semi-definite matrix \citep{horn2012matrix}, we learn that $\lambda_i^R \leq \lambda_i^S \leq \lambda_i^F, i=1,\cdots, n$ for $R \subseteq S \subseteq F$,  where $\lambda_i^R$, $\lambda_i^S$ and $\lambda_i^F$ are the $i$-th largest eigenvalues of matrices as defined in Equation \eqref{eq:asygt} in Section \ref{sec:gt}. Thus, $\bm{\lambda}^S$ is between the upper bound $\bm{\lambda}^F$ and the lower bound $\bm{\lambda}^R$. Then at level $\ell\in [\ell_R, \ell_F]$, we define a ``majorizing vector'' as
\begin{equation}\label{eq:mev}
\hat{\bm{\lambda}}^F_R(\ell) = (\lambda_1^F, \cdots, \lambda_{j_{\ell}-1}^F, \eta(\ell), \lambda_{j_{\ell}+1}^R, \cdots, \lambda_n^R). 
\end{equation}
where $j_{\ell} = \min\{s: \sum_{i=1}^s \lambda^{F|R}_i \geq (\ell- \ell_R)\}$, $\bm{\lambda}^{F|R} = (\lambda^F_1 - \lambda^R_1,\cdots, \lambda^F_n-\lambda^R_n)$ is the pairwise difference of $\bm{\lambda}^F$ and $\bm{\lambda}^R$ and $\eta(\ell) = \lambda_{j_{\ell}}^R + (\ell-\ell_R - \sum_{i=1}^{j_{\ell}-1} \lambda^{F|R}_i)$. For the special case with $j_{\ell} = 1$, we let $\eta(\ell) = \lambda_{j_{\ell}}^R + (\ell-\ell_R)$ and thus $\hat{\bm{\lambda}}^F_R(\ell) = ( \eta(\ell), \lambda_2^R, \cdots, \lambda_n^R)$. It is also obvious that $\hat{\bm{\lambda}}^F_R(\ell_F) = \bm{\lambda}_F$. 

The majorizing vector at level $\ell$ simply takes the first few largest values of the upper bound $\bm{\lambda}^F$ as head and the last few smallest values of the lower bound $\bm{\lambda}^R$ as tail, and connecting them by an $\eta(\ell)$ such that $\hat{\bm{\lambda}}^F_R(\ell)$ is in descending order and its sum is $\ell$ based on Equation \eqref{eq:mev}. Obviously, $\hat{\bm{\lambda}}^F_R(\ell)$ is still bounded by $\bm{\lambda}^F $ and $ \bm{\lambda}^R$, but it majorizes all other eigenvalue vectors at the same level $\ell$. 

We argue that the critical value that is computes by the majorizing vector is maximal among the ones at the same level, in terms of the following theorem in \cite{bock1987inequalities}.
\begin{theorem}\label{tm:main}
Suppose that $\bm{\lambda} \succ \bm{\delta}$. Then there exists an $\alpha_0$ such that,  for $\alpha \leq \alpha_0$, we have
$$c(\bm{\lambda}) \geq c(\bm{\delta}).$$
\end{theorem}
A proof of Theorem \ref{tm:main} is in \cite{bock1987inequalities}, we only change notations. Note that the theorem only holds for small significance level, i.e.\ $\alpha \leq \alpha_0$. 

By combining Theorem \ref{tm:main} and the definition of the majorizing vector, we define the maximal critical value as:
\begin{equation}
c_{max}(\ell) = c(\hat{\bm{\lambda}}^F_R(\ell)), 
\end{equation}
which has the property described in the following lemma,

\begin{lemma}\label{lm:cmax}
For $\alpha \leq \alpha_0$, $c_{max}(\ell)$ satisfies \eqref{eq:cmax_ii} for all $S$ with $R\subseteq S \subseteq F$.
\end{lemma}

%% something about the crossing points of these two pdf 
In above lemma we may see that the validity of $c_{max}$ depends on $\alpha_0$, which has to be sufficiently large for Lemma \ref{lm:cmax} to be useful. \cite{diaconis1990bounds} compared the tail probabilities of $\sum\limits_{i=1}^n \lambda_i \chi^2_1$ and $\sum\limits_{i=1}^n \delta_i \chi^2_1$ with $\bm{\lambda} \succ \bm{\delta}$. They conjectured that the corresponding cumulative distribution functions (cdf) of $\sum\limits_{i=1}^n \lambda_i \chi^2_1$ and $\sum\limits_{i=1}^n \delta_i \chi^2_1$ cross exactly once, implying that $\alpha_0$ would be far from 0 or 1. However, their conjecture was disproved by \cite{yu2017unique} who showed that the two cdfs cross an odd number of times (but sometimes more than once). However, the cdf of $\sum\limits_{i=1}^n \lambda_i \chi^2_1$ will be always below that of $\sum\limits_{i=1}^n \delta_i \chi^2_1$ after the last crossing point, as Theorem \ref{tm:main} claims. The value of $\alpha_0$ in the paper is exactly tail probability corresponding to the last crossing point. Usually, practitioners would like to take significance level $\alpha = 5\%$, which requires $ 5\% \leq \alpha_0 $. We tested this in the real data examples, where we find that $\alpha_0$ is typically in the range of $25\%$ to $ 30\%$; more details are available in the supplementary file. Because of the model robustness of closed testing mentioned in Section \ref{sec:ct}, we only need $5\% \leq \alpha_0 $ for hypothesis $H_N$.

In the toy example in Figure \ref{fig:tc}, given the upper bound $\bm{\lambda}^F$ and the lower bound $\bm{\lambda}^R$ with $R=\{3\}$, the $c_{max}(\ell)$ line and the exact critical values $c_S$ for all $H_S$ are presented as dashed line and triangle points. It is clear that $c_{max}(\ell)$ is above all exact critical values. In addition to avoiding the exponentially many critical value computations, we further note that calculating $\hat{\bm{\lambda}}^F_R(\ell)$ for all possible levels only requires calculation of eigenvalues $\bm{\lambda}^F$ and $\bm{\lambda}^R$ once. This significantly reduces the computing time especially for large matrices (i.e.\ large $n$).

\subsection{Sure or Unsure Outcomes}\label{sec:unsure}

With everything set in place, we check whether $H_R$ can be rejected by the single-step shortcut via checking if the minimum test statistic line is above the maximal critical value line. If $g_{min}(\ell) \geq c_{max}(\ell), \ell \in [\ell_R, \ell_F]$, $H_R$ is surely rejected by the closed testing procedure based on Lemma \ref{lm:ss}. For example, $H_{\{3\}}$ in Figure \ref{fig:tc} is rejected by closed testing at level 5\%, as the $g_{min}$ line is totally above the $c_{max}$ line indicating that all hypotheses corresponding to the supersets of $\{3\}$ are rejected. Otherwise, we turn to the conclusion that $H_R$ cannot be rejected by closed testing so as to guarantee the FWER control. We thus summarize the ``reject'' and ``not reject'' rule of the single-step shortcut as:
\[
\text{Reject } H_R \text{ if } g_{min}(\ell) \geq c_{max}(\ell), \forall \ell \in [\ell_R, \ell_F] \text{ and do not reject } H_R \text{ otherwise.}
\]

In fact, however, uncertainty regarding to the ``not reject'' output of the single-step shortcut will occur when $g_{min}(\ell) < c_{max}(\ell)$ for some $\ell$. In this case, we further check the exact test statistics and exact critical values for all sets in $\mathcal{B}^F_R$, the bottommost points defined in Section \ref{sec:gmin}; see the supplement for technical definition. If there exists a set $B_i \in \mathcal{B}^F_R$ such that $g_{B_i} < c_{B_i}$, we are certain that closed testing does not reject $H_R$. For example when testing $H_{\{2\}}$ in Figure \ref{fig:tc2}, we find that Globaltest does not reject $H_{\{24\}}$ and $H_{\{245\}}$ so that $H_{\{2\}}$ cannot be rejected by closed testing. On the other hand, if $g_{B_i} \geq c_{B_i}$ for all $ B_i \in \mathcal{B}^F_R$, we will be inconclusive about the ``not reject'' of $H_R$ by closed testing, which is the case in the top subfigure of Figure \ref{fig:bb}, where we cannot determine that $H_{\{1\}}$ is not rejected by closed testing. 

\begin{figure}[H]
\centering
\includegraphics[scale=0.68]{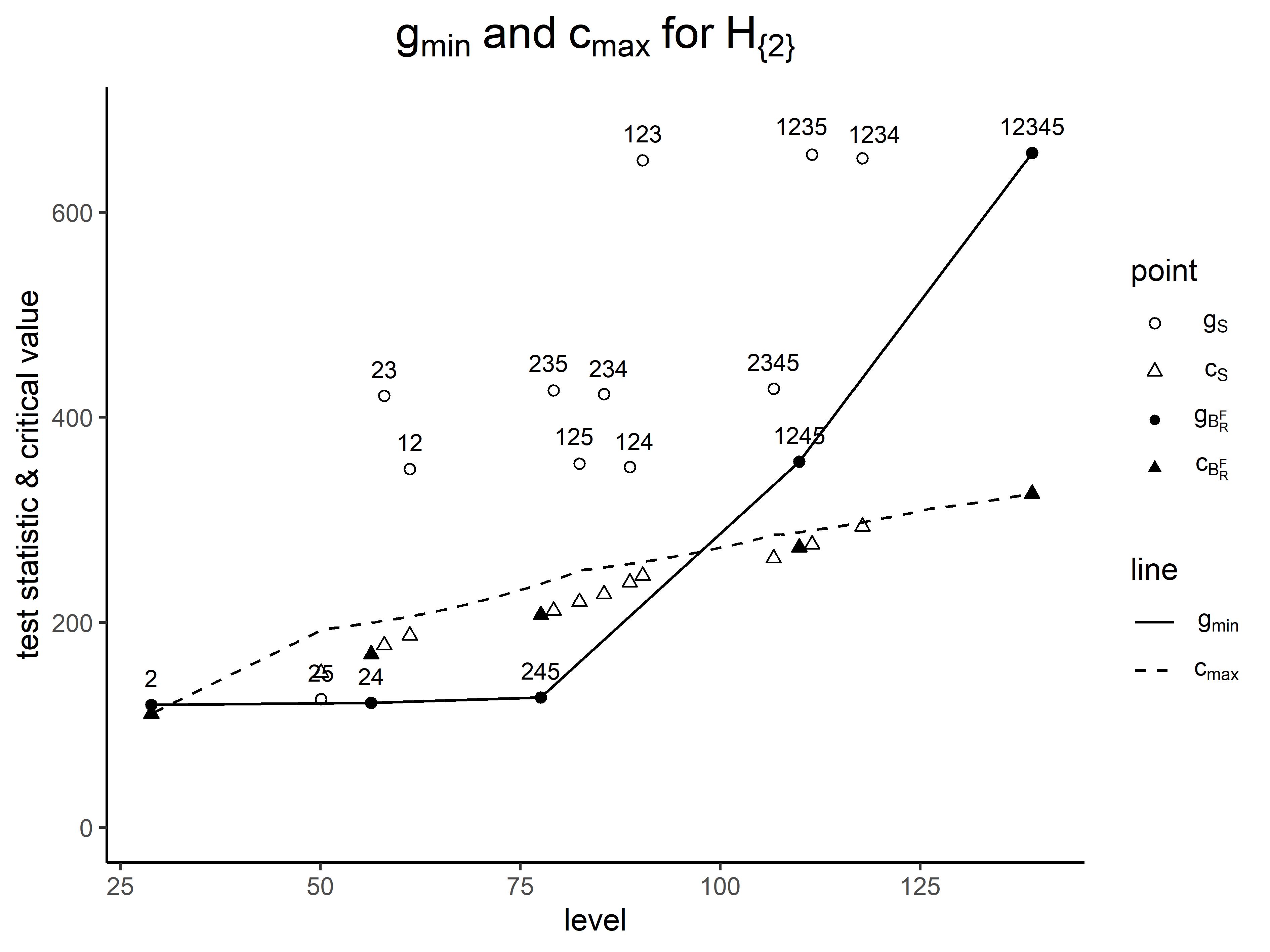}
\caption{Single-step shortcut for testing $H_{\{2\}}$. Filled circles and triangles represent the exact test statistics and critical values for $B_i \in \mathcal{B}^F_R$. }
\label{fig:tc2}
\end{figure}

To further improve the computational efficiency of the single-step shortcut, we derive a fast algorithm to check if $g_{min}(\ell) \geq c_{max}(\ell), \forall \ell \in [\ell_R, \ell_F]$ in the supplementary material, with an illustration on the toy example mentioned above. Regarding to the potential uncertainties of the single-step shortcut, we describe a novel improvement to it in the following section.

\section{Iterative Shortcut}\label{sec:it}
Clearly, the single-step shortcut is approximate in the sense that it gives at most the same rejections as the full closed testing procedure, but possibly fewer because we might get unsure outcomes. Next, we investigate how we can make it exact. If an unsure outcome is obtained from the single-step shortcut, we turn to the branch and bound algorithm of \cite{landautomatic1960}, which is a commonly used tool for solving NP-hard optimization problems. 

The branch and bound algorithm consists of two principles: a recursive branching rule that partitions the search space into smaller sub-spaces and a bounding rule that used for tracking the optimization in the sub-spaces and pruning those sub-spaces that it can prove will not contain an optional solution. \cite{westfall2007multiple} has introduced its application in closed testing with max T test. We show in this paper how it can be used to reduce the conservativeness of the single-step shortcut at the expense of an increased computational burden.

Suppose that the single-step shortcut produces an unsure outcome when testing $H_R$. This means that the $g_{min}$ line crosses the $c_{max}$ line in the space of $\{S: R\subseteq S \subseteq F\}$ and $g_{B_i} \geq c_{B_i}$ for all $B_i \in \mathcal{B}^F_R$. In terms of the branch and bound algorithm, we first split $\{S: R\subseteq S \subseteq F\}$ into two disjoint sub-spaces by distinguishing whether or not $u \in F\setminus R$ is included: $\mathbb{S}^1 = \{S: R\subseteq S \subseteq F\setminus \{u\} \}$ and $\mathbb{S}^2 = \{S: R \cup \{u\} \subseteq S \subseteq F\}$. We clearly explain in the supplement how to determine $u$ together with the technical constructions of $g_{min}$. Secondly, we compute the $g_{min}$ line and the $c_{max}$ line separately for each subspace. If $g_{min} \geq c_{max}$ in both subspaces, we stop branching and conclude that $H_R$ is rejected by closed testing, as all $H_S$ that are split into two subspaces are all rejected by Globaltest.  If there exists a subspace, say $\mathbb{S}^1$, such that $g_{B_i} < c_{B_i}$ for some $B_i \in \mathcal{B}_{R}^{F\setminus \{u\} }$, we can also stop branching and conclude that closed testing does not reject $H_R$. Otherwise, there exists a subspace for which uncertainty remains. In this case, we can repeat the above steps until we get sure outcomes or we exceed the allotted computational capacity.

Illustration of the branch and bound algorithm can be seen in Figure \ref{fig:bb}, where we are unsure to reject $H_1$ or not by the single-step shortcut. After splitting the full space into two: $\{S: R\subseteq S \subseteq F\setminus \{3\} \}$ and $\{S: R \cup \{3\} \subseteq S \subseteq F\}$, updated $g_{min}$ and $c_{max}$ lines are calculated in each subspace. It is clear that $g_{min} \geq c_{max}$ holds in each subspace so that $H_1$ is definitely rejected by closed testing. 

\begin{figure}[H]
\centering
\includegraphics[scale=0.78]{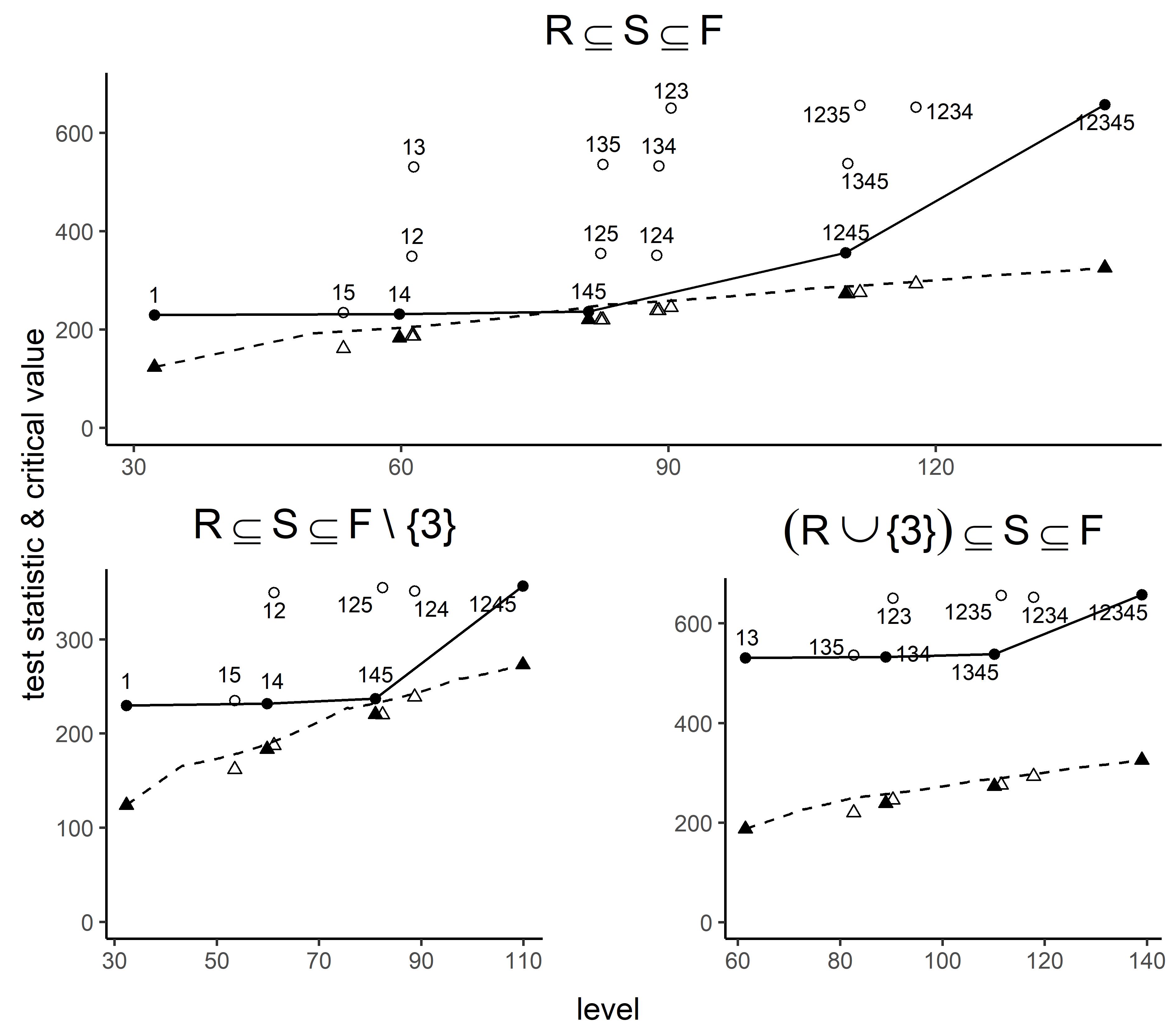}
\caption{Iterative shortcut rejects $H_{\{1\}}$.}
\label{fig:bb}
\end{figure}

By embedding the branch and bound algorithm, we iterate the single-step shortcut to achieve an iterative shortcut. We outline the iterative shortcut as an algorithm in the supplementary file. The theoretical property of the iterative shortcut is investigated by the following lemma, a proof of which can be seen in the supplement:
\begin{lemma}\label{lm:is}
Closed testing rejects $H_R$ only if it is rejected by the iterative shortcut.
\end{lemma}

Obviously, the stopping rule of the iterative shortcut depends on the number of iterations: how many times we iterate the single-step shortcut. The more iterations, the more power we gain. Moreover, we allow user to prespecify the number of iterations to save computation time but without sacrificing FWER control. If we apply the shortcut long enough so that no unsure outcomes left, the full closed testing solution will be obtained.

As described in Section \ref{sec:ct}, $\mathcal{X}$ is the rejection set of closed testing. Let $\mathcal{X}_{d}$ be the rejection set of the iterative shortcut with $d$ iterations prespecified. Specifically, $\mathcal{X}_{0}$ is the set of rejections by the single-step shortcut and $\mathcal{X}_{\infty} = \lim_{d \to \infty}\mathcal{X}_{d}$ is the asymptotic rejection set of iterative shortcut. Obviously, $\mathcal{X}_{\infty} = \mathcal{X}$ holds as there are totally $2^{m-r}$ sub-spaces created after infinite branching, which are exactly corresponding to all possible supersets of $R$. Thus, we have the following convergence property of the iterative shortcut:
\begin{lemma}\label{lm:is_prop}
$\mathcal{X}_{0} \subseteq \mathcal{X}_{d} \subseteq  \mathcal{X}_{\infty}  = \mathcal{X}$.
\end{lemma}

There is clearly a trade-off between computing time and coming close to the full closed testing when applying branch and bound algorithm. We can trade time for power, i.e.\ more iterations, closer to the full closed testing procedure but also more burdensome. But we also allow to limit the number of iterations to reduce the computation time in practice.

\section{Applications on Metabolomics Data}

To investigate the power property of closed testing with Globaltest (CTGT), we apply it to four real metabolomics data sets, whose role on regulatory pathways of human pathophysiology, ranging from aging to disease, has been highlighted. The detailed information of the four data sets are listed in Table \ref{tab:4dsets}. ``Eisner'' is downloaded from MetaboAnalyst \citep{xia2015metaboanalyst}, ``Bordbar'', ``Taware'' and ``Al-Mutawa'' are studies published in MetaboLights \citep{haug2020metabolights}. To normalize the data containing only positive values, such as Eisner, we use the logarithm transformation with base 2. while for the data containing negative values or zeros, such as ``Bordbar'', ``Taware'' and ``Al-Mutawa'', we use the generalized logarithm \citep{durbin2002variance}. 

To be able to compare CTGT with DAG, Structured Holm (SH) and Focus Level (FL), we chose pathways of interest a priori: the union of pathways from Biocyc, KEGG, SMPDB and WikiPathways. The first three annotation vocabularies are provided by``MBROLE'' and the last is generated by ``rWikiPathways'' \citep{slenter2017wikipathways}. We include individual metabolites as single pathways, after removing missing values and filtering out lowly expressed metabolites. Information of pathways, including the total number of pathways to be tested, the mean size and the maximal size, is presented in Table \ref{tab:pathwaydb}. In addition, we consider closed testing with Simes's test (CTST) as a competitor, which is post hoc as well.     

\begin{table}[!htbp]
\caption{Information of four metabolomics data sets. } \label{tab:4dsets}
\centering
    \begin{tabular}{l l c l  }
Data set & Response & Cases/Controls &  Reference \\
    \toprule
       Eisner  & cachexic muscle loss & 47/30  & \cite{Eisner2011} \\
       Bordbar & stimulated lipopolysaccharide & 6/6  &\cite{bordbar2012model} \\ 
       Taware  & head and neck carcinoma & 53/39  &   \cite{taware2018volatilomic} \\  
       Al-Mutawa & hypoxic preconditioning & 25/19 &  \cite{al2018effects} \\         
   \bottomrule
  \end{tabular}
\end{table}

\begin{table}[!htbp]
\caption{ Summary of pathways per data set. } \label{tab:pathwaydb}
\centering
    \begin{tabular}{l c c c c c}
  Data set & Samples & Metabolites & Pathways & Mean size & Max size \\
    \toprule
       Eisner    &  77  &   63  & 187    & 2    & 17 \\
       Bordbar   &  12  &  51   & 250    & 3    & 38 \\ 
       Taware    &  92  &  47   & 61     & 2    & 12 \\  
       Al-Mutawa &  44  &  261  & 760    & 10   & 173 \\         
   \bottomrule
  \end{tabular}
\end{table}

We describe the total number of rejections per method together with the number of shared rejections of any two methods as a lower triangular table for each data set separately, see in Table \ref{tab:EBTA}. Within each data set, $\text{CTGT}$ represents the exact closed testing with Globaltest, i.e.\ the iterative shortcut without unsure outcomes, which can be achieved by setting a large enough number of iterations, we set 20000 in this analysis. This does not mean we have to iterate 20000 times, since the iterative shortcut stops whenever there is no unsure outcomes left. For example, no unsure outcomes are left after iterating 2645 times for Eisner.

It is shown in Table \ref{tab:EBTA} that the CTGT method discovers more pathways than its competitors for data sets Eisner and Bordbar. Especially for Bordbar with only 12 samples but 51 metabolites, the small sample size could weaken the effects of metabolites to some extent, thereby leading to good power of Globaltest. Furthermore, the small sample size influences CTGT method less than the Bonferroni-based methods and CTST due to their reliance on very small tail probabilities. We note that not all results are in favor of CTGT, for example for Taware and Al-Mutawa in Table \ref{tab:EBTA}. In Taware, the low dimensionality and relatively few small-size pathways to be tested make DAG, SH and FL powerful. Remark that we chose the pathways of interest a priori, but only CTST and CTGT retain type I error control if pathways are chosen post hoc. In other words, the more pathway databases included, the less powerful DAG, SH and FL become.

\begin{table}[!htbp]
\caption{Number of rejections per method on the diagonal and number of shared rejections of any two methods under the diagonal for Eisner, Bordbar, Taware and Al-Mutawa.} \label{tab:EBTA}
\centering
    \begin{tabular}{l c c c  c c  l c c c  c c  }
    \toprule
    \multicolumn{6}{c}{Eisner} & \multicolumn{6}{c}{Bordbar} \\  \cline{2-6} \cline{ 8 - 12}
                            &  $\text{CTGT} $    &  CTST         &    SH     &     FL     &     DAG   & & $\text{CTGT}$     &   CTST          &    FL             &     DAG    &     SH       \\              
      %\midrule

  & 144                &                 &           &            &                  & & 248                &                   &                   &             &                           \\

  & 130                 &   139          &           &            &                   &  & 244                &  244              &                   &             &            \\

  & 101                 &   102          &  102       &           &                    &     & 105                &   105             &    105            &             &            \\ 
  & 89                 &   89            &   88       &    89     &                   &        & 62                 &   62              &     40            &   62         &            \\
  & 88                 &   88             &   87      &     84     &   88             &      & 0                  &   0               &     0             &   0          &    0       \\

    \multicolumn{6}{c}{Taware}& \multicolumn{6}{c}{Al-Mutawa} \\  \cline{2-6} \cline{ 8 - 12}  
  &     DAG                &  SH                 &             CTST     &      FL             &  $\text{CTGT}$    & &  CTST            &          SH     &        DAG       &     FL      &   $\text{CTGT}$       \\           
 & 32                 &                      &                        &                     &        && 704               &                &                   &            &          \\
 & 32                 &   32                 &                         &                    &        && 693               &  693           &                   &            &             \\
 & 32                 &   32                &  32                     &                     &        && 683               &   681          &  683              &            &           \\
 & 30                 &   30                &  30                     &   30                &        && 653               &   653          &  651              & 653         &              \\
 & 24                 &   24                 &  24                    &   24                &    27  && 585               &   583          &  585              &   580       &    586    \\       
   \bottomrule
  \end{tabular}
\end{table}

To gain insight into the rejected pathways per method, we take Eisner as an illustration.
We sort all rejected pathways by the rank of their size as shown in Figure \ref{fig:duri}, where we can see that the rejected pathways by CTGT are mainly large ones. This implies that CTGT is preferable for testing large-size pathways. CTST, DAG, SH and FL are, in contrast, powerful for small-size pathways, especially when there are many features with strong signal. However, they are sensitive to the total number of hypotheses and the sample size. 

\begin{figure}[!htbp]
\centering
\includegraphics[scale=0.68]{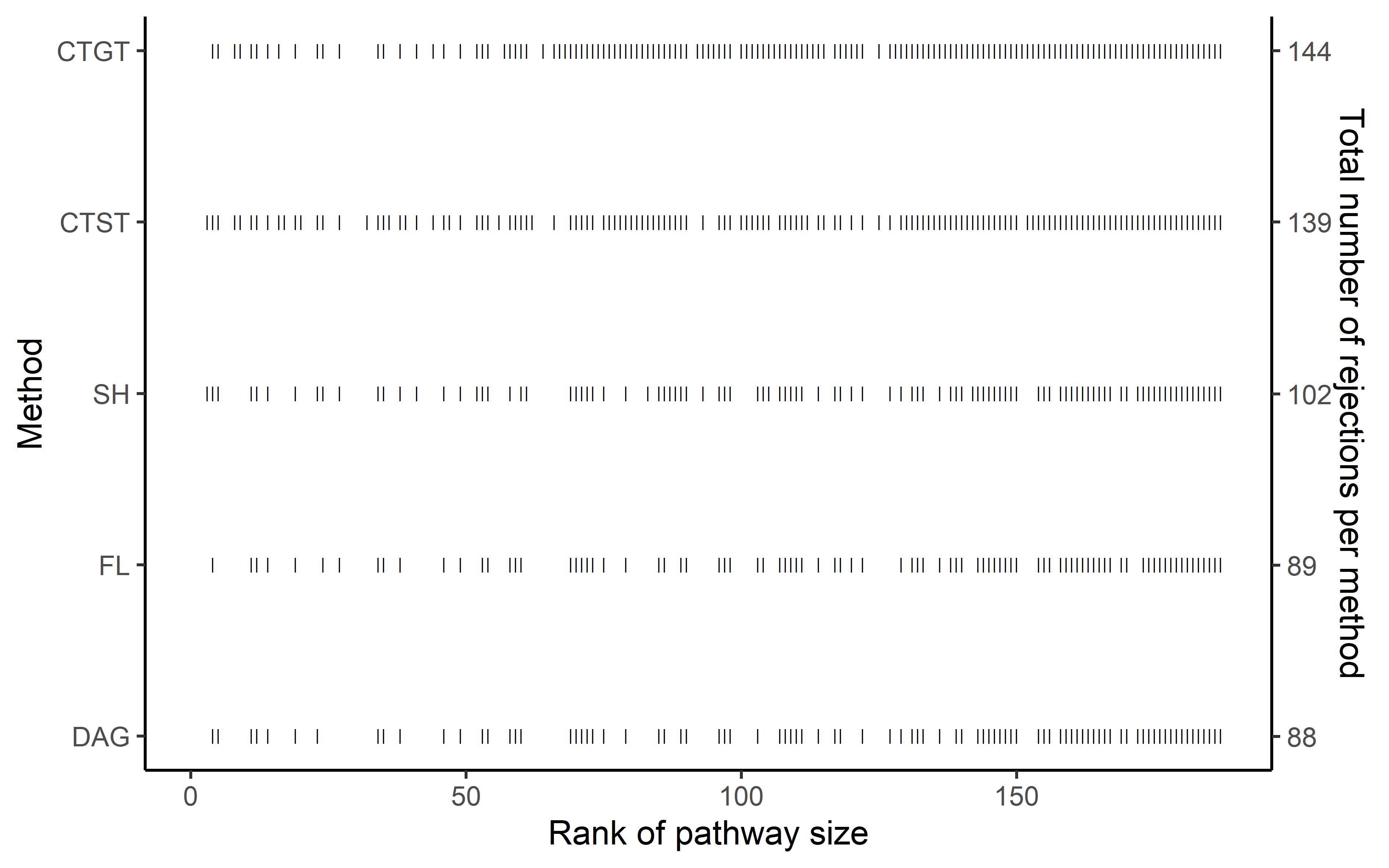}
\caption{Rejected pathways per method for Eisner.  }
\label{fig:duri}
\end{figure}

As demonstrated in Lemma \ref{lm:is_prop} that CTGT method gets more powerful with more iterations, and when it has zero unsure outcomes it corresponds to the exact closed testing procedure. We take Al-Mutawa as an example to illustrate Lemma \ref{lm:is_prop}. We first apply the single-step shortcut to Al-Mutawa, which generates 58 unsure outcomes. Given a sequence of increasing numbers as specified iterations, we show in Figure \ref{fig:time} that more iterations, less unsure outcomes but more certain outcomes (i.e.\ 57 are surely rejected and 1 is surely not rejected) and accordingly longer computing time. Also note that the gain of power is mainly from the initial 5000 iterations, suggesting that unnecessary computing burden can be avoid by the iterative shortcut without much loss of power. The computing time is calculated based on a Dell PowerEdage M620 machine with CPU Intel E5-2697.

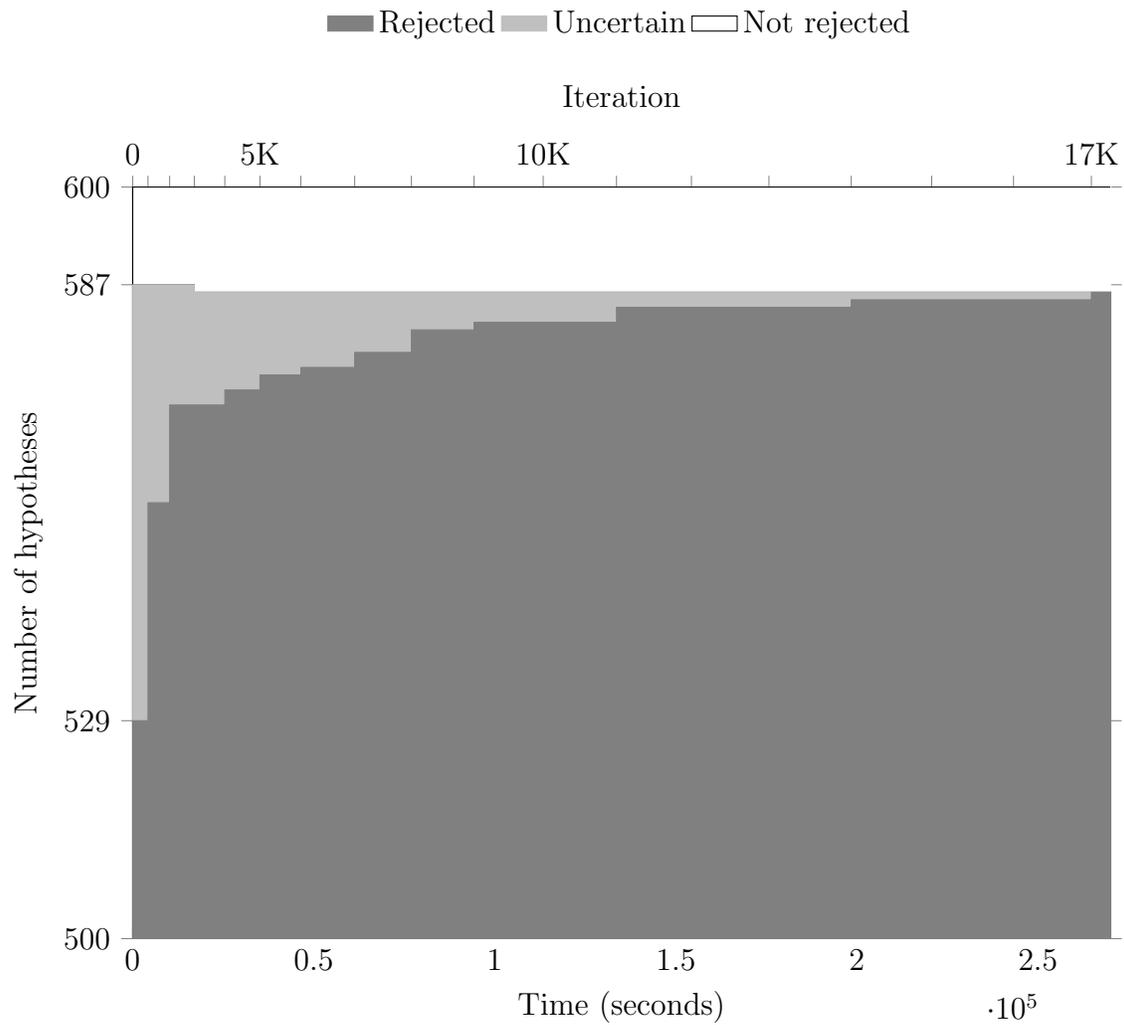
\begin{figure}[!htbp]
\centering
\begin{tikzpicture}
\pgfplotsset{every axis/.style={ymin=0, xmin=0}}
	\begin{axis}[const plot,
	height=10cm,
width=13cm,
	scale only axis,
	xtick = {0,50000,100000,150000,200000,250000},
	ytick align=outside,
	xmin=0, xmax = 270000,
	ymin=500, ymax = 600,
	ytick={500,529,587,600},
    yticklabels={500,529,587,600},
	enlarge x limits = false,
	stack plots=y,
	ylabel={Number of hypotheses},
	xlabel={Time (seconds)},
	cycle list={%
		{gray, fill=gray},%
		{lightgray, fill=lightgray},
		{black, fill=white}},
         legend style={
         draw=none,
			area legend,
			at={(0.5,1.25)},
			anchor=north,
			legend columns=3}]]
	\addplot table[x=time, y=reject, col sep=comma]  {A.csv} \closedcycle;
	\addplot table[x=time, y=unsure, col sep=comma]  {A.csv} \closedcycle;
	\addplot table[x=time, y=accept, col sep=comma]  {A.csv} \closedcycle;
	\legend{Rejected,Uncertain,Not rejected};
	\end{axis}
       \begin{axis}[scale only axis,
	xmin=0, xmax = 270,
	ymin=500, ymax = 600,
		height=10cm,
        width=13cm,
        axis x line*=right, 
        axis y line = none,
        xtick ={0.049,4.255,10.24,17.04,25.474,35.173,46.473,61.329,76.994,94.291,113.375,133.618,154.308,175.712,198.435,220.656,243.243,264.74},
        xlabel= Iteration,
        xtick align=outside,
	x label style={at={(0.5,1.15)},anchor=north},
xticklabels={0,,,,,5K,,,,,10K,,,,,,,17K}]
        ]
	%\addplot[draw=none] table[x=time, y=reject, col sep=comma]  {A.csv} \closedcycle;
	\end{axis}
\end{tikzpicture}
	\caption{Number of rejected, unsure and not rejected hypotheses with increasing computing time and iterations.}
\label{fig:time}
\end{figure}

\section{Discussion}
We have proposed a novel multiple testing procedure based on closed testing with Globaltest, with main applications on metabolomics annotation databases, for example pathway databases. Closed testing is a powerful and robust procedure for FWER control. Globaltest is a powerful test for pathway analysis, especially when there are many weak features in the data. It can adapt to the actual correlation structure of features. Moreover, it is valid in high-dimensional settings, which are becoming increasingly common especially in the area of biomedical science. 

Our method is post hoc, i.e.\ it controls FWER for all possible feature sets so that it allows the choice of feature sets of interest to be made after seeing the data. This is the most distinguishing property of our method. When multiple annotation databases are of interest, type I error remains under control, even when databases are chosen in a data-driven way. Still, the new method has comparable power to competing methods even when databases are specified beforehand. 

To reduce the computational burden of closed testing, we have derived a shortcut in two steps: a single-step shortcut and an iterative shortcut. Based on these shortcuts, a hypothesis may be rejected by closed testing without carrying out the full procedure. The single-step shortcut might still have inconclusive outcomes. It is complemented by the iterative shortcut using the branch and bound algorithm on pathways with unsure outcome. We allow the iterative shortcut to stop at any point while retaining FWER control. With more iterations the method will get closer to the full closed testing procedure but will also be more time-consuming. Once there are no uncertain outcomes left, the iterative shortcut becomes the full closed testing procedure. We have implemented both shortcuts in the \texttt{R} package \texttt{ctgt}.

A potential limitation of the method is that it is only valid for small significance level $\alpha$, i.e.\ $\alpha \leq \alpha_0$, but we have shown in the supplement that $\alpha_0$ is usually greater than most values of $\alpha$ that are used in practice. 

Another advantage of closed testing with Globaltest is that no tests are performed at reduced significance levels, what the Bonferroni-based methods do. This may improve the finite sample performance of the method, since the convergence of the type I error of tests to their nominal size tends to be faster for large $\alpha$ levels than for small $\alpha$ levels \citep{feller1957introduction}.

\bibliographystyle{Chicago}
\bibliography{ctgtbiblio}

\end{document}